\documentclass[11pt, a4paper]{article}
 
\usepackage{microtype}
\usepackage{amsmath,amssymb,amsthm,booktabs,color,doi,enumitem,graphicx,url}
\usepackage{stmaryrd}
\usepackage{amsfonts}
\usepackage{hyperref}
\usepackage[margin=1in]{geometry}

\newtheorem{claim}{Claim}
\newtheorem{theorem}{Theorem}
\newtheorem{corollary}[theorem]{Corollary}

\theoremstyle{definition}
\newtheorem{definition}{Definition}

\theoremstyle{remark}

\newcommand{\ld}{\mathsf{LD}}
\newcommand{\loglcp}{\mathsf{logLCP}}

\newcommand{\lh}{\mathsf{LH}}

\newcommand{\llcp}{\mbox{\sf LogLCP}}
\newcommand{\np}{\mbox{\sf NP}}
\newcommand{\mso}{\mbox{\sf MSO}}
\newcommand{\opt}{\mbox{\sc opt}}
\newcommand{\adm}{\mbox{\sc adm}}
\newcommand{\stt}{\mbox{\sc spanning tree}}
\newcommand{\nta}{\mbox{\sc nontrivial automorphism}}
\newcommand{\ts}{\mbox{\sc travelling salesman}}
\newcommand{\ksat}{\mbox{\sc qbf-sat}}
\newcommand{\vck}{\mbox{\sc cycle-vc-dimension}} 
\newcommand{\mst}{\mbox{\sc mst}} 
\newcommand{\hc}{\mbox{\sc hamiltonian cycle}} 

\newcommand{\idx}{\mbox{\rm index}} 
\DeclareMathOperator{\Q}{Q}

\DeclareMathOperator{\id}{id}

\newcommand{\co}[1]{\operatorname{co-\Lambda}_{#1}}

\hypersetup{
    colorlinks=true,
    linkcolor=black,
    citecolor=black,
    filecolor=black,
    urlcolor=[rgb]{0,0.1,0.5},
    pdfauthor={Laurent Feuilloley, Pierre Fraigniaud, and Juho Hirvonen},
    pdftitle={A hierarchy of local decision},
}

\newenvironment{myabstract}
               {\list{}{\listparindent 1.5em%
                        \itemindent    \listparindent
                        \leftmargin    0pt
                        \rightmargin   0pt
                        \parsep        0pt}%
                \item\relax}
               {\endlist}

\newenvironment{mycover}
               {\list{}{\listparindent 0pt
                        \itemindent    \listparindent
                        \leftmargin    0pt
                        \rightmargin   0pt
                        \parsep        0pt}%
                \raggedright
                \item\relax}
               {\endlist}

\begin{document}

\vspace*{2ex}
\begin{mycover}
{\huge\bfseries A Hierarchy of Local Decision\footnote{The first and second authors received additional support from ANR project DISPLEXITY, and from INRIA project GANG.}\par}
\bigskip
\bigskip

\textbf{Laurent Feuilloley}

\nolinkurl{laurent.feuilloley@liafa.univ-paris-diderot.fr}
\medskip

{\small Institut de Recherche en Informatique Fondamentale (IRIF), \\ CNRS and University Paris Diderot, France \par}
\bigskip

\textbf{Pierre Fraigniaud}

\nolinkurl{pierre.fraigniaud@liafa.univ-paris-diderot.fr}
\medskip

{\small Institut de Recherche en Informatique Fondamentale (IRIF), \\ CNRS and University Paris Diderot, France \par}
\bigskip

\textbf{Juho Hirvonen}

\nolinkurl{juho.hirvonen@aalto.fi}
\medskip

{\small Helsinki Institute for Information Technology HIIT, \\ Department of Computer Science, Aalto University, Finland\par}
\medskip

\title{}

\end{mycover}
\bigskip
\begin{myabstract}
\noindent\textbf{Abstract.} We extend the notion of \emph{distributed decision} in the framework of distributed network computing, inspired by recent results on so-called \emph{distributed graph automata}. We show that, by using distributed decision mechanisms based on the interaction between a \emph{prover} and a \emph{disprover}, the size of the certificates distributed to the nodes for certifying a given network property can be drastically reduced. For instance, we prove that minimum spanning tree can be certified with $O(\log n)$-bit certificates in $n$-node graphs, with just one interaction between the prover and the disprover, while it is known that certifying MST requires $\Omega(\log^2n)$-bit certificates if only the prover can act. The improvement can even be exponential for some simple graph properties. For instance, it is known that certifying the existence of a nontrivial automorphism requires $\Omega(n^2)$ bits  if only the prover can act. We show that there is a protocol with two interactions between the prover and the disprover enabling to certify nontrivial automorphism with $O(\log n)$-bit certificates. These results are achieved by defining and analysing a \emph{local hierarchy} of decision which generalizes the classical notions of \emph{proof-labelling schemes} and \emph{locally checkable proofs}.
\end{myabstract}

\thispagestyle{empty}
\setcounter{page}{0}
\newpage


\section{Introduction}

\subsection{Context and Objective}

This paper is tackling the long-standing issue of characterizing the power of local computation in the framework of distributed network computing~\cite{Peleg00}. Our concern is the ability to design \emph{local} algorithms, defined as distributed algorithms in which every node of a network (i.e., every computing entity in the system) can compute its output after having consulted only nodes in its vicinity. That is, communications proceed along the links of the network, and, in a local algorithm, every node must output after having exchanged information with nodes at constant distance only. A \emph{construction} task consists, for the nodes of a network $G=(V,E)$ where each node~$u$ is given an input $x(u)$, to collectively and concurrently compute a collection  $y(u)$, $u\in V$, of  individual outputs, such that $(G,x,y)$ satisfies some property characterizing the task to be solved. For instance, the minimum-weight spanning tree (MST) task consists, given the weights $x(u)$ of all the incident edges of every node~$u$, in computing a subset~$y(u)$ of edges incident to~$u$ such that the set $\{y(u), u\in V\}$ forms a MST in~$G$. Similarly, the maximal independent set (MIS) task consists of computing $y(u)\in \{0,1\}$, $u\in V$, such that the set $\{u\in V: y(u)=1\}$ forms an MIS. It is an easy observation that the MST task cannot be solved locally as the weights of far-away edges may impact the output of a node. In a seminal result Linial showed that the same is true for MIS~\cite{Linial92}: there is no local algorithm for constructing an MIS, even on an $n$-node ring. Nevertheless, there are many construction tasks that can be solved locally, such as approximate solutions of NP-hard graph problems (see, e.g.,~ \cite{FloreenHKKMS11,KuhnMW04,LenzenOW08,LenzenW08}). In general it is Turing-undecidable whether or not a construction task can be solved locally~\cite{NaorS95}. 

Interestingly, the Turing-undecidability result of Naor and Stockmeyer~\cite{NaorS95} concerning  the locality of construction tasks holds even if one restricts the question to properties that can be locally decided. A distributed \emph{decision} task~\cite{FraigniaudKP13} consists, given an input $x(u)$ to every node in a network~$G$, to decide whether $(G,x)$ satisfies some given property. An instance is accepted by a distributed algorithm if and only if every node individually accepts (i.e., every node~$u$ outputs $y(u)=\mbox{true}$). For instance, proper colouring can easily be  decided locally by having each node merely comparing its colour with the ones of its neighbours. On the contrary, deciding whether a collection of edges defined by $\{x(u), u\in V\}$ forms a MST is not possible locally (in fact, even separating paths from cycles is not possible locally). Similarly to the  sequential computing setting, there are strong connections between the construction variant of a task and the ability to locally decide the legality of a given candidate solution for the same task, as illustrated by, e.g., the derandomization results in~\cite{FeuilloleyF15,NaorS95}, and the approximation algorithms in~\cite{SarmaHKKNPPW12}. These connections have motivated work focusing on the basic question: what can be decided locally? This paper is aiming at pushing further our current knowledge on this question.

Two specific lines of work have motivated our approach of local decision in this paper. The first line of research is related to the notion of \emph{proof-labelling schemes} introduced by Korman et al.~\cite{KormanKP10}, who showed that while not all graph properties can be decided locally, they can all be \emph{verified} locally, with the help of local \emph{certificates} provided by a prover. Unfortunately, there are natural graph properties (e.g., the existence of a non-trivial automorphism) which require $\Omega(n^2)$-bit certificates to be verified by any local distributed algorithm~\cite{GoosS11}. G\"{o}\"{o}s and Suomela introduced the more practical class $\llcp$ of all graph properties that can be verified using certificates of size $O(\log n)$ bits~\cite{GoosS11}, i.e., merely the size required to store the identities of the nodes. The class $\llcp$ contains non locally decidable properties such as hamiltonicity and non-bipartiteness. $\llcp$  even contains graph properties that are not in $\np$. Also, all existential-$\mso$ graph properties are shown to be in $\llcp$. 

The second line of research which motivated our approach is the study of \emph{distributed graph automata}. In particular, \cite{Reiter15} recently proved that an analogue of the polynomial hierarchy, where sequential polynomial-time computation is replaced by distributed local computation, turns out to coincide with $\mso$. However, while this result is important for our understanding of the computational power of finite automata, the model does not quite fit with the standard model of distributed computing aiming at capturing the power of large-scale computer networks (see, e.g., \cite{Peleg00}). Indeed, on the one hand, the model in~\cite{Reiter15} is somewhat weaker than desired, by assuming a finite-state automata at each node instead of a Turing machine, and by assuming anonymous computation instead of the presence of unique node identities. On the other hand, the very same model is also stronger than the standard model, by assuming a decision-making mechanism based on an arbitrary mapping from the collection of all node states to $\{\mbox{true},\mbox{false}\}$. Instead, the classical distributed decision mechanism is based on the  logical conjunction of the individual decisions. This is crucial as this latter decision mechanism provides the ability for every node rejecting the current instance to raise an alarm, and/or to launch a recovery procedure, without having to collect all of the sindividual decisions. 

In this paper, our objective is to push further the study initiated in~\cite{GoosS11} on the $\llcp$ class, by adopting the approach of~\cite{Reiter15}. Indeed, $\llcp$ can be seen as the first level $\Sigma_1$ of a \emph{local hierarchy} $(\Sigma_k,\Pi_k)_{k\geq 0}$,  where $\Sigma_0=\Pi_0=\ld$, the class of properties that can be locally decided~\cite{FraigniaudKP13}, and, for $k\geq 1$, $\Sigma_k$ is the class of graph properties for which there exists a local algorithm $A$ such that,  for every instance $(G,x)$,  
\[
(G,x) \;\mbox{is legal} \iff \exists \ell_1 \forall \ell_2 \exists \ell_3 \dots Q \ell_k : A(G,x,\ell_1,\ell_2,\dots,\ell_k) \; \mbox{accepts}
\]
with $k$ alternations of quantifiers, and where $Q$ is the universal quantifier if $k$ is even, and the existential quantifier otherwise. ($\Pi_k$ is defined similarly as $\Sigma_k$, but starting with a universal quantifier). The $\ell_i$'s are called \emph{labelling functions}, assigning a label $\ell_i(v)\in\{0,1\}^*$ to every node~$v$, such that, for every node~$v$, $|\ell_i(v)|=O(\log n)$ in $n$-node networks. Our aim is to analyze the local hierarchy in the general context of distributed network computing~\cite{Peleg00}, where each node has an identity which is unique in the network, every node has the computational power of a Turing machine, and where the acceptance of an instance by an algorithm is defined as the logical conjunction of the individual decisions of the nodes. 

\subsection{Our Results}

We study a hierarchy $(\Sigma_k,\Pi_k)_{k\geq 0}$ of local decision which represents a natural extension of proof-labelling scheme, as well as of locally checkable proof, with succinct certificates (i.e., of size $O(\log n)$ bits). In addition to its conceptual interest, this hierarchy might have some practical impact. Indeed, any level $k$ of the hierarchy can be viewed as a game between a \emph{prover} and a \emph{disprover}, who play in turn by alternating $k$ times. Roughly, on legal instances, the prover aims at assigning distributed certificates responding to any attempt of the disprover to demonstrate that the instance is illegal, and vice-versa on illegal instances. The referee judging the correctness of the collection of certificates produced by the players is a local distributed algorithm. For instance, the extensively studied class $\Sigma_1$ includes problems whose solutions are such that their legality can be certified by a prover using distributed certificates. Instead, the class $\Pi_2$  includes problems whose solutions are such that their legality can be certified by a prover against any other candidate solution provided by a disprover, both using distributed certificates. 

We show that many problems have succinct proofs in the hierarchy. Actually, climbing up the hierarchy enables to reduce drastically the size of the certificates. For instance, we show a quadratic improvement for MST, which requires locally checkable proofs of $\Omega(\log^2n)$ bits, while MST stands at the second level of our hierarchy. That is, there is a $\Pi_2$-protocol for MST using distributed certificates of $O(\log n)$ bits. For graph properties such as nontrivial automorphism, the improvement can even be exponential in term of certificate size, by relaxing the verification from locally checkable proofs with $\Omega(n^2)$ bits proofs to $\Sigma_3$ (with $O(\log n)$ bits proofs). More generally,  many natural optimization problems are on the second level of our hierarchy. On the other hand, we also show that there are simple  (Turing-computable) languages outside the local hierarchy. This latter property illustrates the impact of insisting on compact $O(\log n)$-bits certificates: there are graph properties that cannot be locally certified via a finite number of interactions between a prover and a disprover using succinct certificates. 

In addition, we prove several results regarding the structure of the hierarchy.  In particular, we show that if the hierarchy collapses partially at any level, then it collapses all the way down to that level. On the other hand, we prove that the hierarchy does not collapse to the first level (i.e., the first and second levels are distinct). Distributed decision is naturally asymmetric, that is, reversing the individual decision of the algorithm at each node does not correctly reverse the global decision of the algorithm. As a consequence, it is not necessarily the case that $\mbox{co-}\Sigma_k=\Pi_k$, and vice-versa. However, we show that one additional level of quantifiers is always sufficient to reverse a decision (i.e., to decide the complement of a language). Finally, we show that, for every graph property at the intersection of a level-$k$ class and the complement of this class, there is a protocol deciding that property at level $k$ with \emph{unanimous} decision, for both legal and illegal instances. 

\subsection{Related Work}

Several forms of ``local hierarchies'' have been investigated in the literature, with the objective of understanding the power of local computation, or for the purpose of designing verification mechanisms for fault-tolerant computing. In particular, as we already mentioned, \cite{Reiter15} has investigated the case of \emph{distributed graph automata}, where nodes are anonymous finite automata, and where the decision function is a global interpretation of the all the individual outputs of the nodes. In this context, it was proved that the local hierarchy is exactly captured by the MSO formulas on graphs. 

The picture is radically different in the framework in which the computing entities are Turing machines with pairwise distinct identities, and where  the decision function is the logical conjunction of all the individual boolean outputs. In \cite{FraigniaudKP13}, the authors investigated the local hierarchy in which the certificates must not depend on the identity-assignment to the nodes. Under such \emph{identity-oblivious} certificates, there are distributed languages outside~$\Sigma_1$. However, all languages are in the probabilistic version of $\Sigma_1$, that is, in $\Sigma_1$ where the correctness of the verification is only stochastically guaranteed with constant probability. In \cite{FraigniaudHK12}, it is proved that $\Sigma_1$ is exactly captured by the set of distributed languages that are closed under lift. (A configuration $(G',x')$ is a $t$-lift of a configuration $(G,x)$ if there is an input-preserving mapping from $V(G')$ to $V(G)$ which preserves the $t$-neighbourhood of the nodes in these graphs). Interestingly,  in the same framework as~\cite{FraigniaudKP13} but where the decision function is a global interpretation of the all the individual outputs, instead of the logical conjunction of individual boolean outputs, \cite{ArfaouiFIM14,ArfaouiFP13} proved that the local hierarchy collapses to $\Sigma_1$. Also, in the same framework as~\cite{FraigniaudKP13}, but where the certificates may depend on the identity assignment, all distributed languages are in $\Sigma_1$ (see \cite{KormanKP10}). 

The literature about  the local hierarchy in the context of Turing machine computation tend to show that all languages are at the very bottom of the hierarchy. However, \cite{GoosS11} proved that, to be placed in $\Sigma_1$, there are distributed languages on graphs (e.g., the existence of a nontrivial automorphism) which require to exchange certificates of size $\Omega(n^2)$ bits among neighbours, which is enough to trivially decide any problem. Similarly, \cite{KormanK07,KormanKP10} has proved that certifying Minimum-weight Spanning Tree (MST) requires to exchange certificates on $\Theta(\log^2n)$ bits, which can be costly in networks with limited bandwidth, i.e., under the CONGEST model~\cite{Peleg00}. In \cite{KormanKM15}, it is proved that the size of the certificates for MST can be decreased to $O(\log n)$ bits, but to the expense of $O(\log n)$ rounds of communication. Recently, \cite{MorFP15} has proved that the amount of communication between nodes (but not necessarily the size of the certificates) for verifying general network configurations can be exponentially decreased if using randomization, and~\cite{ForsterLSW16} analyzed in depth the certificate size for $s$-$t$ connectivity and acyclicity. 

It is also worth mentioning the role of the node identities in distributed decision. For instance, after noticing that the identities are leaking information to the nodes about the size of the network (e.g., at least one node has an ID at least $n-1$ in $n$-node network), it was recently proved that restricting the  algorithms to be identity-oblivious reduces the ability to decide languages locally in $\Sigma_0$ (see~\cite{FraigniaudGKS13}), while this is not the case for $\Sigma_1$ (see~\cite{FraigniaudHK12}). Recently, \cite{FraigniaudHS15} characterized the ``power of the IDs'' in local distributed  decision. In \cite{EmekSW14}, the authors discussed what can be computed in an anonymous networks, and showed that the answer to this question varies a lot depending on the commitment of the nodes to their first computed output value, i.e., whether it is revocable or not. In the context of local decision, the output is assumed to be irrevocable.

In general, we refer to~\cite{Suomela13} for a recent survey on local distributed computing, and we refer to \cite{FraigniaudRT13,FraigniaudRT14} for distributed decision in the context of asynchronous crash-prone systems with applications to runtime verification, and to~\cite{ArfaouiF14} for distributed decision in contexts where nodes have the ability to share non classical resources (e.g., intricate quantum bits).


\section{Local Decision} \label{sec:preliminaries}

Let $G = (V,E)$ denote an undirected graph, where $V$ is the set of nodes, and~$E$ is the set of edges. The subgraph induced by nodes at distance (i.e., number of hops) at most $t$ from a node $v$ is denoted by $B_G(v,t)$. All graphs considered in this paper are assumed to be connected (for non connected graphs, our results apply separately to each connected components). The number of nodes in the graph is denoted by $n$. In every graph $G=(V,E)$, each node $v\in V$ is assumed to have a name from the set $\{1, \dots, N\}$, denoted by $\id(v)$, where $N$ is polynomial in $n$. In other words, all identities are stored on $O(\log n)$ bits. In a same graph, all names are supposed to be pairwise distinct. 

\vspace*{-2ex}

\subparagraph{Distributed languages.} 

A \emph{distributed language} $L$ is a set of pairs $(G,x)$, where $G$ is a graph and $x$ is a function that assigns some local input $x(v)$ to each node~$v$. We assume that all inputs $x(v)$ are polynomial in~$n$, and thus can be stored locally on $O(\log n)$ bits. The following are typical examples of distributed languages:
\begin{itemize}[noitemsep]
	\item \textsc{3-colouring}: $(G,x)$ such that $x$ encodes a proper 3-colouring of $G$;
	\item \textsc{3-colourability}: graphs that  can be properly 3-coloured;
	 \item \textsc{nta}: graphs with a nontrivial automorphism;
	\item \textsc{planarity}: planar graphs. 
\end{itemize}
The \emph{complement} $\bar{L}$ of a distributed language $L$ is defined as the set $\bar{L} = \{ (G,x) : (G,x) \notin L \}$. For instance, the complement of \textsc{3-colouring} is \textsc{non-3-colouring}, consisting of all pairs $(G,x)$ such that $x$ is not a proper $3$-colouring of $G$.

\vspace*{-2ex}

\subparagraph{Labellings.} 

A \emph{labelling} $\ell$ is a function $\ell \colon V(G) \to \{0,1\}^*$, assigning a bit string to each node. If, for every graph $G$ and every node $v\in V(G)$, $\ell(v)\in\{0,1\}^k$, we say that the labelling $\ell$ is of size $k$. In this paper, we are mostly interested in labellings of logarithmic size in the number of nodes in the input graph.

\vspace*{-2ex}

\subparagraph{Local algorithms.} 

We use the standard \textsf{LOCAL} model of distributed computing~\cite{Peleg00,Linial92}. In this model each node $v \in V(G)$ is a computational entity that has direct communication links to other nodes, represented by the edges of $G$. Every node runs the same algorithm. In this paper, all algorithms are deterministic. Nodes communicate with each other in synchronous communication rounds. During each round, every node is allowed to (1) send a message to each of its neighbours, (2) receive a message from each of its neighbours, and (3) perform individual computation. At some point every node has to halt and produce an individual output. The number of communication rounds until all nodes have halted is the \emph{running time} of an algorithm.

A  \emph{local}  algorithm is a distributed algorithm $A$ for which there exists a constant $t$ such that, for every instance $(G,x)$, the running time of $A$ in $(G,x)$ is at most $t$. Since the most a node can do in $t$ communication rounds  is to gather all the information available in its local neighbourhood $B_G(v,t)$, a local algorithm $A$ can be defined as a (computable) function from all possible labelled local neighbourhoods to some output set. Given an ordered set $\bar{\ell} = (\ell_1, \ell_2, \dots, \ell_k)$ of labellings, for some $k\geq 0$, and given an instance $(G,x)$, we denote by $A(G,v,x,\bar{\ell})$ the output of node $v$ in algorithm $A$ running in $G$ with input function $x$ and labelling $\bar{\ell}$.

\vspace*{-2ex}

\subparagraph{Local decision.} 

In \emph{distributed decision}, the output of each node $v$ corresponds to its own individual decision. That is, each node either \emph{accepts} or \emph{rejects}. Globally, the instance $(G,x)$ is accepted  if and only if every node accepts individually. In other words, the global acceptance is subject to the logical conjunction of all the individual acceptances. For the sake of simplifying the presentation, $A(G,v,x,\bar{\ell}) = 1$ (resp., $A(G,v,x,\bar{\ell}) = 0$) denotes the fact that $v$ accepts (resp., rejects) in an execution of algorithm $A$ on $(G,x)$ labelled with $\bar{\ell}$. We say that $A$ accepts if $A(G,v,x,\bar{\ell}) = 1$ for every node $v\in V(G)$, and rejects otherwise. We will use the shorthand $A(G,x,\bar{\ell}) = 1$ to denote that $\forall v \in V(G), A(G,v,x,\bar{\ell}) = 1$, and $A(G,x,\bar{\ell}) = 0$ to denote that $\exists v \in V(G), A(G,v,x,\bar{\ell}) = 0$.

The first class in the local hierarchy considered in this paper is \emph{local decision}, denoted by $\ld{}$. A language $L$ is in $\ld{}$ if there exists a local algorithm $A$, such that for all graphs $G$, and all possible inputs $x$ on $G$, we have:
\[
(G,x) \in L \iff A(G,x) \text{ accepts.} 
\]
As an example, deciding whether $x$ is a 3-colouring of $G$ is in $\ld{}$, but deciding whether $G$ is 3-colourable is not. Note that $\ld{}$ does not refer to any labellings. The algorithm $A$ runs solely on graphs $G$ with possible inputs to the nodes. 

Finally, for sake realism, we assume that the nodes are Turing machines, and we consider only languages that are decidable in the centralized setting. This differs from previous works, in particular from the study of $\llcp{}$ by Göös and Suomela \cite{GoosS11} where no assumption is made about the computational power of the nodes. For simplicity, as this change in the model only affects theorem \ref{thm:outsideLH} (where we prove that a language is not in our hierarchy), we will abuse notation and write that $\llcp$ is the same as the first level of the local hierarchy. 

\vspace*{-2ex}

\subparagraph{Example: certifying spanning trees.} 

In a graph $G$, a \emph{spanning tree} can be encoded as a distributed data-structure $x$ such that, for every $v\in V(G)$, $x(v)$ encodes the identity of one of $v$'s neighbours (its parent in the tree), but one node $r$ for which $x(r)=\bot$ (this node is the root of the tree). Deciding whether $x$ is a spanning tree of $G$ is not in $\ld{}$. However, a spanning tree can be \emph{certified} locally as follows (see~\cite{ItkisL94}).  Given a spanning tree $x$ of $G$ rooted at node $r$, a \emph{prover} assigns label $\ell(v)=(\id(r),d(v))$ to each node~$v$, where $d(v)$ is the distance of $v$ to the root $r$ in the spanning tree $x$. Such a label is on $O(\log n)$ bits. The verification algorithm $A$ at node $v$ checks that $v$ agrees on $\id(r)$ with all its neighbours, and that $d(x(v))=d(v)-1$. If both tests are passed, then $v$ accepts, otherwise it rejects. It follows that Algorithm $A$ accepts if and only if $x$ is a spanning tree of $G$. Indeed, if $x$ is not a spanning tree of $G$, there is no way to assign ``fake'' labels  to the nodes so that all nodes accepts. This ability to certify spanning trees is a simple but powerful tool that will be used throughout the paper.  

\section{The Local Hierarchy} \label{sec:local-hierarchy}

\newcommand{\spt}{\mbox{\sc spanning-tree}}

We generalize the various classes of distributed decision from previous work into a hierarchy of distributed decision classes, in a way analogous to the polynomial hierarchy. In particular, our class $\Sigma_1^{\ld{}}$ is equivalent to the class  $\loglcp{}$ introduced by G\"{o}\"{o}s and Suomela~\cite{GoosS11} (up to the hypothesis on the computational power of the nodes).

\subsection{Definition}

We first define an infinite hierarchy  $\{(\Sigma_k^{\ld{}})_{k\ge 0},(\Pi_k^{\ld{}})_{k\ge 0}\}$ of classes. For the sake of simplifying the notation, each of these classes is now abbreviated in $\Sigma_k$ or $\Pi_k$. Informally, each class can be defined by a game between two players, called the  \emph{prover} and the \emph{disprover}. Both players are given a language $L$ and an instance $(G,x)$. In $\Sigma_k$ (resp., $\Pi_k$), with $k>0$, the prover (resp., disprover) goes first, and assigns an $O(\log n)$-bit label to each node. Then, the players alternate,  assigning $O(\log n)$-bit labels to each node in turn, until $k$ labels $\ell_1, \ell_2, \dots, \ell_k$ are assigned. A language $L$ is in the corresponding class if there is a local algorithm $A$, and a prover-disprover pair such that, given $(G,x)$, for every set of labels that the disprover assigns, the prover can always assign labels such that $A$ accepts if and only if $(G,x) \in L$. That is, if $(G,x) \notin L$, then the disprover can always force some node to reject, whatever the prover does. Such a combination local algorithm $A$ and prover-disprover pair is called a \emph{decision protocol} for $L$ in the corresponding class. Equivalently, we define $\ld{} = \Sigma_0 = \Pi_0$, and, for $k>0$, $\Sigma_k$ is defined as the set of languages $L$ for which there exists $c\geq 0$,  and  a local algorithm $A$ such that
\[
	(G,x) \in L \iff \exists \ell_1, \forall \ell_2, \dots, \Q \ell_k, A(G,x,\ell_1, \ell_2, \dots, \ell_k) = 1,
\]
where $\Q$ is the existential (resp., universal) quantifier if $k$ is odd (resp., even), and every label $\ell_i$ is of size at most $c \log n$. The class $\Pi_k$ is defined similarly, except that the acceptance condition is: 
$
	(G,x) \in L \iff \forall \ell_1, \exists \ell_2, \dots, \Q \ell_k, A(G,x, \ell_1, \ell_2, \dots, \ell_k) = 1. 
$

\vspace{-2ex}

\subparagraph{Remark.} 

For both $\Sigma_k$ and $\Pi_k$, the equivalence should hold \emph{for every identity-assignment} to the nodes with identities in $[1,N]$, where $N$ is a fixed function polynomial in~$n$. Indeed, the membership of an instance $(G,x)$ to a language is independent of the identities given to the nodes. On the other hand, the labels given by the prover and the disprover may well depend on the actual identities of the nodes in the graph where the decision algorithm $A$ is run. This is for instance the case of the protocol for certifying spanning trees described in the previous section, establishing that $\spt\in\Sigma_1$. 

\subsection{The odd-even collapsing, and the $\Lambda_k$-hierarchy}

Interestingly, the ending universal quantifier in both $\Sigma_{2k}$ and $\Pi_{2k+1}$ do not help. The class $\Pi_1$ turns out to be just slightly stronger than $\ld{}$. Specifically, we prove the following result.

\begin{theorem}\label{thm:collapse-theorem}
For every $k\geq 1$, $\Sigma_{2k} = \Sigma_{2k-1}$ and $\Pi_{2k+1} = \Pi_{2k}$. Moreover, $\ld{}\subseteq \Pi_1 \subseteq  \ld{}^{\mbox{\footnotesize \rm \#node}}$, that is,  local decision with access to an oracle providing each node with the number of nodes in the graph.
\end{theorem}

\begin{proof}
The result follows from the fact that an existential quantification on labels of size $O(\log n)$ bit is sufficient to provide the nodes with the exact size of the graph using a spanning tree as described in Section~\ref{sec:preliminaries}. 

\begin{claim}\label{claim:spanningtree}
Let $L=\{(G,x): \mbox{for every $v\in V(G)$, $x(v)=|V(G)|$}\}$. We have $L\in\Sigma_1$. 
\end{claim}

To establish the claim, on a legal instance, let $T$ be a spanning tree of $G$, and root $T$ at node $r$. Let us set label $\ell(v)=(\id(r),p(v),s(v))$ where $p(v)$ is the identity of the parent of $v$ in $T$, and  $s(v)$ is the size of the subtree of $T$ rooted at $v$.  The verification proceeds as follows: each node $v$ checks that it agrees on $\id(r)$ and $x(v)$ with all its neighbours in the graph, and that $s(v)=1+\sum_{w \in p^{-1}(v)}s(w)$ where $p^{-1}(v)$ denotes the set of $v$'s children, i.e., all neighbours $w$ of $v$ such that $p(w)=v$. In addition, the root $r$ checks that $s(r)=x(r)$. If all tests are passed, then the node accepts, otherwise it rejects. It follows that this algorithm accepts if and only if $x(v)=n$ for all nodes $v$. This completes the proof of the claim. 

We show how to use this mechanism in the case of $\Sigma_{2k}$, for $k> 0$. Let $L\in \Sigma_{2k}$, and let $A$ be a $t$-round local algorithm such that:
\[
(G,x) \in L \iff \exists \ell_1, \forall \ell_2, \dots, \exists \ell_{2k-1}, \forall \ell_{2k}, A(G,x,\ell_1, \ell_2, \dots, \ell_{2k}) = 1.
\]
Recall that all labellings $\ell_i$, $i=1,\dots,2k$, are of size at most $c \log n$ for some $c\geq 0$. We construct an algorithm $A'$ that simulates $A$ for a protocol that does not need the last universal quantifier on $\ell_{2k}$. The first labelling $\ell'_1$ consists of some correct $\ell_1$ for $A$, with the aforementioned additional label that encodes a spanning tree $x'$ (rooted at an arbitrary node) and the value of the number of nodes in $G$. Regarding the remaining labellings, for each $\ell_{2i-1}$ assigned by the disprover, the prover assigns $\ell_{2i}$ as in the protocol for $A$, ignoring the bits padded to $\ell_1$ for creating $\ell'_1$. After the labellings have been assigned, each node $v$ gathers its radius-$t$ neighbourhood $B_G(v,t)$. Then, it virtually assigns every possible combination of $(c \log n)$-bit labellings $\ell_{2k}(u)$ to each node $u \in B_G(v,t)$, and simulates $A$ at $v$ to check whether it accepts or rejects with this labelling. If every simulation accepts, then $A'$ accepts at $v$,  else it rejects. Since every nodes generate all possible $\ell_{2k}$ labellings in its neighbourhood, we get that 
\[
	(G,x) \in L \iff \exists \ell'_1, \forall \ell_2, \dots, \exists \ell_{2k-1}, A(G,\ell_1, \ell_2, \dots, \ell_{2k-1}) \text{ accepts}~, 
\]
which places $L$ in $\Sigma_{2k-1}$. 
The proof of $\Pi_{2k+1}=\Pi_{2k}$ is similar by using the first existential quantifier (which appears in second position) to certify the number of nodes in the graph. For the case of $\Pi_1$, the nodes use the value of the number of nodes directly provided by the oracle \#node. 
\end{proof}

A consequence of Theorem~\ref{thm:collapse-theorem} is that only of the classes $\Sigma_k$ for odd $k$, and $\Pi_k$ for even $k$, are worth investigating.

\begin{definition}
We define the classes $(\Lambda_k)_{k\ge 0}$ as follows: 
$
\Lambda_k = \left \{ \begin{array}{ll}
 \Sigma_k & \mbox{if $k$ is odd;} \\
 \Pi_k & \mbox{otherwise}.
 \end{array} \right.
$
\end{definition}

In particular, $\Lambda_0 = \Pi_0= \ld{}$. By definition, we get $\Lambda_k \subseteq \Lambda_{k+1}$ for every $k\geq 0$, as the distributed algorithm can simply ignore the first label. 

\begin{definition}
The local hierarchy is defined as $\lh{}=\cup_{k\geq 0} \Lambda_k$.
\end{definition}

\subsection{Complementary classes} \label{ssec:compl-classes}

We define the complement classes $\co{k}$, for $k\geq 0$, as
$
\co{k} = \{ L : \bar{L} \in \Lambda_k \}.
$
Note that, due to the asymmetric nature of distributed decision (unanimous acceptance, but not rejection), simply reversing the individual decision of an algorithm deciding $L$ is generally not appropriate to decide $\bar{L}$. Nevertheless, we show that an additional existential quantifier is sufficient to reverse any decision, implying the following theorem.

\begin{theorem} \label{thm:cotheorem}
For every $k\geq 0$, $\co{k} \subseteq \Lambda_{k+1}$.
\end{theorem}

\begin{proof}
The proof uses a spanning tree certificate to reverse the decision, in a way similar to the proof that the complement of $\ld{}$ is contained in $\loglcp{}$ (i.e., according to our terminology, $\co{0} \subseteq \Lambda_1$) due to G\"{o}\"{o}s and Suomela~\cite{GoosS11}. Let $L \in \Lambda_k$, and let $A$ be a $t$-round local  algorithm deciding $L \in \Lambda_k$ using labels on at most $c\log n$ bits. We  construct an algorithm $A'$ which simulates $A$, but uses an additional label $\ell_{k+1}$ to reverse the decisions made by $A$. Let us assume that $k$ is even (as it will appear clear later, the proof is essentially the same for $k$ odd).  We have that
\[	
(G,x) \in L \iff \forall \ell_1, \exists \ell_2, \dots, \exists \ell_{k}, A(G,x,\ell_1, \ell_2, \dots, \ell_k) = 1,
\]
with all labels $\ell_i$'s of size at most $c \log n$ for some constant $c\geq 0$. In Algorithm $A'$, the prover and the disprover essentially switch their roles. From the above, we have 
\[
 	(G,x) \notin L \iff \exists \ell_1, \forall \ell_2, \dots, \forall \ell_k, \exists v \in G, A(G,v,x,\ell_1, \ell_2, \dots, \ell_k) = 0.
\]
The prover for $A'$ always follows the disprover for $A$, and can always pick labellings $\ell_1, \ell_3, \dots, \ell_{k-1}$ such that there is a rejecting node if and only if $(G,x)\notin L$. In the protocol for $A'$, the prover sets $\ell_{k+1}$ to be a spanning tree rooted at one such rejecting node~$v$. Every other node $u \neq v$ simply checks that $\ell_{k+1}$ constitutes a proper encoding of a spanning tree, and rejects if not. If all nodes $u \neq v$ accept, then $\ell_{k+1}$ is indeed a proper spanning tree, and it only remains to check that $v$ rejects in $A$. To this end, the node $v$ designated as the root of the spanning tree encoded by $\ell_{k+1}$ gathers all labellings in its radius-$t$ neighbourhood, and computes $A(G,x,v,\ell_1, \ell_2, \dots, \ell_k)$. If $A$ rejects at $v$, we set $A'$ to accept at $v$, and, otherwise, we set $A'$ to reject at $v$.

As discussed in Section~\ref{sec:preliminaries}, the spanning tree can be encoded using $O(\log n)$ bits. All labellings $\ell_1, \ell_2, \dots, \ell_k$ have size at most  $c \log n$, therefore all labels of $A'$ are of size at most $c' \log n$ for some $c'\geq c$. The protocol is correct, as a rejecting node exists in $A$ if and only if $(G,x) \notin L$, and $A'$ correctly accepts in this case. If $(G,x) \in L$, then we have that, for every choice the prover can make, the disprover can always choose its labellings so that $A$ accepts. Thus, if the spanning tree $\ell_{k+1}$ is correct, the root of that tree will indeed detect that it is an accepting node in $A$, and so reject in $A'$.
\end{proof}

\begin{corollary}
For every $k\geq 0$,  $\co{k} \subseteq \co{k+1}$, and $\Lambda_k \subseteq \co{k+1}$.
\end{corollary}

\begin{proof}
If $L \in \co{k}$, then, by definition, $\bar{L} \in \Lambda_k$, and thus also $\bar{L} \in \Lambda_{k+1}$, which implies that $L \in \co{k+1}$. If $L \in \Lambda_k$, then $\bar{L} \in \co{k}$, and thus, by Theorem~\ref{thm:cotheorem}, we get that $\bar{L} \in \Lambda_{k+1}$, which implies that $L \in \co{k+1}$. 
\end{proof}

The following theorem shows that, for every $k\geq 0$, and every language $L$ in $\Lambda_k \cap \co{k}$, there is an algorithm deciding $L$ such that an instance $(G,x)\in L$ is accepted at all nodes, and an $(G,x)\notin L$ is rejected at all nodes.

\begin{theorem} \label{thm:intersection}
Let $k\geq 1$, and let $L\in \Lambda_k \cap \co{k}$. Then there exists a local algorithm $A$ such that, for every instance $(G,x)$, and for every $v\in V(G)$, 
\[
(G,x) \in L \iff \left \{ \begin{array}{ll}
\forall \ell_1, \exists \ell_2,\forall \ell_3, \dots, \exists \ell_k, A(G,v,x,\ell_1,\dots,\ell_k) = 1 & \mbox{if $k$ is even}\\
\hspace*{4.5ex} \exists \ell_1, \forall \ell_2, \dots, \exists \ell_k, A(G,v,x,\ell_1,\dots,\ell_k) = 1 & \mbox{otherwise}
\end{array}\right.
\]
\end{theorem}

\begin{proof}
Assume first that $k$ is even. Since $L\in \Lambda_k \cap \co{k}$, there exist two local algorithms $B$ and $B'$ such that
\begin{equation} \label{eq:strong-accept}
(G,x) \in L \iff \forall \ell_1, \exists \ell_2, \dots, \exists \ell_k, \forall v, B(G,v,x,\bar{\ell}) = 1,
\end{equation}
and
\begin{equation} \label{eq:strong-reject}
(G,x) \notin L \iff \forall \ell_1, \exists \ell_2, \dots, \exists \ell_k, \forall v, B'(G,v,x,\bar{\ell}) = 1.
\end{equation}
Now we construct the following protocol for deciding $L$ in an unanimous manner. If $(G,x) \in L$, then the prover assigns the labels as in~\eqref{eq:strong-accept}. Instead, if $(G,x) \notin L$, then the prover  assigns the labels as in~\eqref{eq:strong-reject}. In addition, the first bit of $\ell_k$ tells which algorithm the nodes should use, with $0$ for $B$ and $1$ for~$B'$.

Now, the decision algorithm $A$ proceeds as follows. For each node $v$, if $v$ and all neighbours of $v$ have the same flag $B$ or $B'$, then $v$ simulates $B$ (with the first bit of $\ell_k$ ignored), and outputs $B(G,v,x,\ell_1,\dots,\ell_k)$. Conversely, $v$ and all the neighbour of $v$ have the same flag $B'$, then $v$ simulates $B'$ and outputs $B'(G,v,x,\ell_1,\dots,\ell_k)$. Finally, if the neighbourhood of $v$ contains both flags $B$ and $B'$, then that node outputs the value of its own flag (i.e., 0 or 1). In other words, any $\ell_k$-labelling with non-consistent flag is rejecting, though not unanimously so.

It remains to check that $A$ behaves correctly when the flag in $\ell_k$ is consistent. By construction, from~\eqref{eq:strong-accept} and~\eqref{eq:strong-reject}, and from the way $A$ uses $B$ and $B'$, we get that the left to right implication in the statement of the theorem is satisfied: there is always a way to label to graph so that \emph{yes}-instances are accepted everywhere, and \emph{no}-instances are rejected everywhere. 

For the other direction, let $f(v)$ denote the flag-bit of the label $\ell_k$ at node $v$. We consider the two cases of whether $(G,x)\in L$ or not. First, assume that $(G,x) \in L$, but $f(v) = B'$ for all nodes $v$. Since every node is simulating $B'$ and reversing its decision, it follows from~\eqref{eq:strong-reject} that
\[
(G,x) \in L \iff \exists \ell_1, \forall \ell_2, \dots, \forall \ell_k, \exists v, B'(G,v,x,\ell_1, \ell_2, \dots,\ell_k) = 0.
\]
That is, if $(G,x) \in L$, then the disprover can force \emph{some} node to accept, even if all nodes are consistently running the wrong algorithm, $B'$.

Conversely, assume that $(G,x) \notin L$, but $f(v)=B$ for all nodes $v$. Similarly to the previous case, since every node is simulating $B$, it follows from~\eqref{eq:strong-accept} that
\[
(G,x) \notin L \iff \exists \ell_1, \forall \ell_2, \dots, \forall \ell_k, \exists v, A(G,v,x,\ell_1, \ell_2, \dots, \ell_k) = 0,	
\]  
That is, if $(G,x) \notin L$, then the disprover can force \emph{some} node to reject. 

The case for $k$ odd is similar. 
\end{proof}

In Theorem~\ref{thm:optimization} in the next section, we shall see several example of languages in $\Lambda_1 \cap \co{1}$, in relation with classical optimization problems on graphs. By Theorem~\ref{thm:intersection}, all of these languages can be decided unanimously. 

\subsection{Separation results}

From the previous results in this section, we get that the local hierarchy $\lh{}=\cup_{k\geq 0} \Lambda_k$ has a typical ``crossing ladder'' as depicted on Figure~\ref{fig:lh-structure}.

\begin{figure}[htb]
	\centering
	\includegraphics[width=0.6\textwidth]{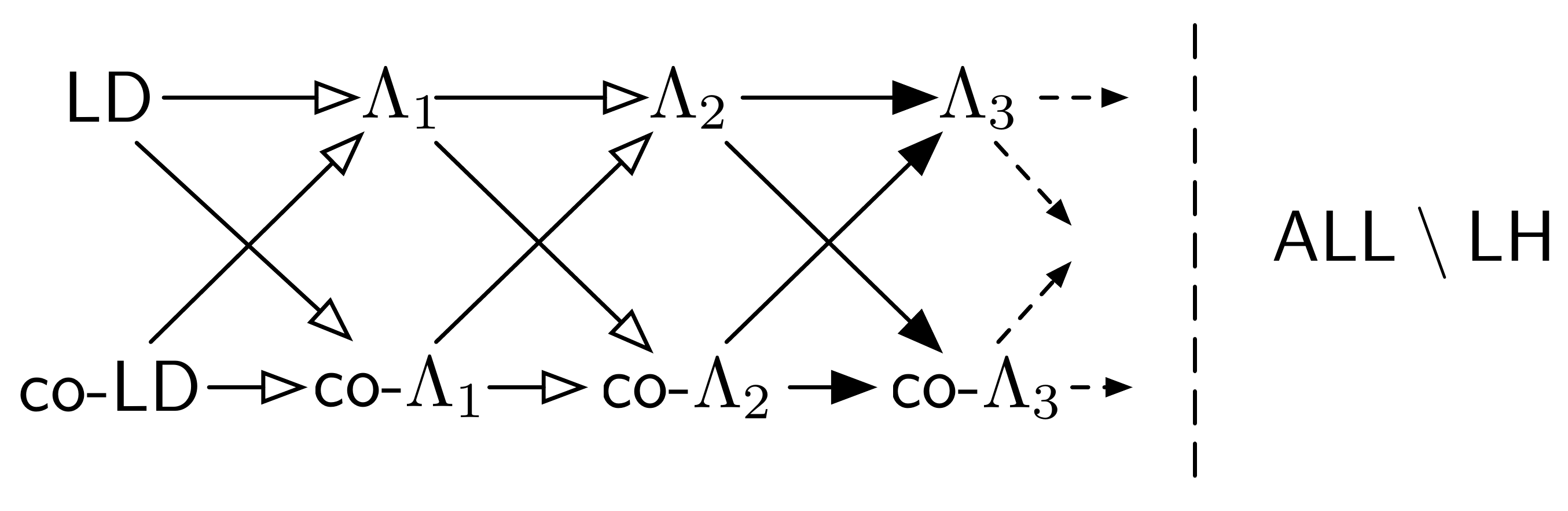}
	\caption{\sl Structure of the local hierarchy. Arrows indicate inclusions, while hollow-headed arrows indicate strict inclusions.}
	\label{fig:lh-structure}
\end{figure}

In addition, we can show that some of the inclusions are strict. Indeed,  it is known for long that $\ld{}$ is strictly included in $\Lambda_1$ (for instance, \textsc{2-colourability} $\in \Lambda_1 \setminus \ld{}$). Also, $\Lambda_0 \cup \co{0}$ is strictly included in $\co{1}$. Indeed, for instance,  \textsc{non-3-colourability} $\in \co{1} \setminus (\Lambda_0 \cup \co{0})$. Therefore, all inclusions between $\ld$ and co-$\ld$ and the classes at the first level are strict. Moreover, it is known~\cite{GoosS11} that \textsc{non-3-colourability}~$\notin\Lambda_1$, implying that \textsc{3-colourability}~$\notin \co{1}$. On the other hand, both languages are in $\Lambda_2$, by application of Theorem~\ref{thm:cotheorem}. As a consequence, both are also in $\co{2}$. Therefore, all inclusions between the classes at the first and second levels are strict. 

For $k\geq 2$, separating the classes at the $k$th level from the classes at the next level appears to be not straightforward. In particular, all classical counting arguments used to separate the three first levels (i.e., levels~0, 1, and~2) fail. On the other hand, we show that if $\Lambda_k=\Lambda_{k+1}$ for some $k$, then $\lh$ collapses to the $k$th level.

\begin{theorem} \label{thm:collapse}
If there exists $k\geq 0$ such that $\Lambda_k = \Lambda_{k+1}$, then $\Lambda_i = \Lambda_k$ for all $i > k$, that is, $\lh$ collapses at the $k$th level.
\end{theorem}

\begin{proof}
Let us assume for the purpose of contradiction, that there exists $k\geq 0$ such that $\Lambda_k= \Lambda_{k+1} \neq \Lambda_{k+2}$. Let $L \in \Lambda_{k+2} \setminus \Lambda_{k+1}$. Let is assume that $k$ is odd (as it will appear clear later, the proof for $k$ even is similar). Since $L \in \Lambda_{i+2}$, there exists a local algorithm $A$ such that 
\[
(G,x) \in L \iff \exists \ell_1, \forall \ell_2, \dots, \exists \ell_{k+2}, A(G,x,\ell_1, \ell_2, \dots \ell_{k+2}) = 1.
\]
We then define the language $\tilde{L}$ as
\[
(G, (x,\ell)) \in \tilde{L} \iff \forall \ell_1, \exists \ell_2, \dots, \exists \ell_{k+1} A(G,x,\ell, \ell_1, \dots, \ell_{k+1}) = 1.
\]
By this definition, we get $\tilde{L} \in \Lambda_{k+1}$. Now, since $\Lambda_k = \Lambda_{k+1}$, we get that there exists a local algorithm $B$ such that
\[
(G, (x,\ell)) \in \tilde{L} \iff \exists \ell_1, \forall \ell_2, \dots, \exists \ell_k, B(G,x,\ell,\ell_1, \ell_2, \dots, \ell_k) = 1.
\]	
On the other hand, by definition, $(G,x) \in L$ if and only if there exists $\ell$ such that  $(G,(x,\ell)) \in \tilde{L}$. Using Algorithm $B$, the latter is equivalent to
\[
\exists \ell, \exists \ell_1, \forall \ell_2, \dots, \exists \ell_k, B(G,(x,\ell),\ell_1,\ell_2,\dots,\ell_k) = 1.	
\]
Now, the two existential quantifiers on $\ell$ and $\ell_1$ can be combined into one to get a protocol establishing  $L\in\Lambda_k$, a contradiction.
\end{proof}

Finally, we show that there are languages outside $\lh$. In fact, this result holds, even if we restrict ourselves to languages with inputs 0 or 1 on  oriented paths, i.e., with identity-assignment where nodes are given consecutive ID from 1 to~$n$.  The result follows from the fact that there are ``only'' $2^{2^{O(\log n)}}$ different local algorithms for such $n$-node instances at any fixed level of $\lh$, while there are $2^{2^n}$ different languages on such instances.

\begin{theorem}\label{thm:outsideLH}
There exists a Turing-computable language on 0/1-labelled oriented paths that is outside $\lh$. 
\end{theorem}

\begin{proof}
The proof is in two steps. First, using a counting argument, we show that, for any fixed set of parameters, that is, for every level $k$, every constant $c$ controlling the label size $c\log n$, and every running time $t$, there is language that cannot be recognized by a protocol with such parameters. Then we combine these languages  for various sets of parameters, for building a (Turing-computable) language that cannot be recognized by any protocol of the hierarchy.

\begin{claim}\label{lem:outsideLH}
Let $k$, $c$, and $t$ be non negative integers. There exists an integer $n$, and a language $L=L(n,k,c,t)$ on 0/1-labelled oriented paths that cannot be recognized by a protocol for $\Lambda_k$ running in $t$ rounds using labels of size at most $c\log n$ bits. 
\end{claim}

To establish the claim, notice that an algorithm is simply a mapping from all possible balls (including identifiers, inputs and labels) to binary outputs (accept or reject). On 0/1 inputs, and IDs in $[1,n]$, The number of algorithms for $\Lambda_k$ running in $t$ rounds using labels of size at most $c\log n$ bits  is at most $2^{2^{\beta\log(n)}}$, where $\beta=\beta(k,c,t)$ depends only on $k,c$, and $t$. On the other hand, the number of languages on words of size $n$ is exactly $2^{2^n}$. Let $n$ be such that 
$
2^{2^{\beta(k,c,t)\log n}}<2^{2^n}.
$ 
By the pigeon-hole principle, there exists a languages that cannot be decided by any algorithm for $\Lambda_k$ running in $t$ rounds using labels of size at most $c\log n$ bits. This completes the proof of the claim. 

\medbreak

Let  $m$ be a nonnegative integer, and let $S(n,m)$ be the set of languages on oriented paths with $n$ nodes that cannot be recognized by a protocol with $k=c=t=m$. By Claim~\ref{lem:outsideLH}, for every $m$ there exists $n$ such that $S(n,m)$ is not empty. We strengthen this by observing the following two points. First, if $S(n,m)$ is non empty, then for every $m'<m$, the set $S(n,m')$ is non-empty as well, since the protocol for $m'$ could be simulated with parameter $m$. Second, if $S(n,m)$ is non-empty, then, for every $n'>n$, the set $S(n',m)$ is also non-empty. Indeed, let $L\in S(n,m)$, and let us consider the language $L'$  composed of the set of words in $L$ padded with zeros. If $L'$ has an algorithm, then we could modify this algorithm  to get an algorithm recognizing $L$. 

Let us define $\mu(n)$ as the largest integer $m$ such that $S(n,m)$ is non-empty, and $\nu(m)$ the smallest integer $n$ such that $S(n,m)$ is non-empty. Given $n$ and $m$ such that $S(n,m)\neq \emptyset$, let $L(n,m)$ be the smallest language of $S(n,m)$ according to the lexicographic ordering. Finally, let $L=(\cup_{n\geq 1} L(n,\mu(n))$.

We first show that $L$ is a distributed language, i.e., that it is Turing-computable. We describe the algorithm deciding $L$. The algorithm, given an $n$-bit string $X$, computes $\mu(n)$ by enumerating all $m$'s in increasing order, by trying, for each of them, all local algorithms with parameter $m$, and by checking whether $S(n,m)=\emptyset$. This algorithm eventually finds $m$ such that $S(n,m)=\emptyset$, giving $\mu(n)=m-1$. Then the algorithm computes $L(n,\mu(n))$, and accepts $X$ if and only if $X\in L(n,\mu(n))$.

We complete the proof by showing that $L\notin \lh$.  Suppose, for the sake of contradiction, that $L\in \lh$. Then there exists a local algorithm $A$ deciding  $L\in \Lambda_k$, running in $t$ rounds using labels of size at most $c\log n$ bits, for some $k,c$, and $t$. Let  $m=\max\{k,c,t\}$.  We can transform $A$ to decide $L\in \Lambda_m$, running in $m$ rounds using labels of size at most $m\log n$ bits. 

Let us consider the restriction $L'$ of $L$ on words of size $\nu(m)$. By definition, $L'=L(\nu(m), \mu(\nu(m)))$, and this language cannot be recognized by a local algorithm with parameter $\mu(\nu(m))$. On the other hand, $\mu(\nu(m))\geq m$, and therefore $L'$ cannot be recognized by an algorithm of parameter $m$ either. In particular, $L'$ cannot be recognized by $A$, a contradiction. Therefore $L\notin \lh$. 
\end{proof}


\section{Positive results} \label{sec:positive-results}

In this section, we precisely identify the position of some relevant problems for distributed computing in the local hierarchy.

\subsection{Optimization problems}

Given an optimization problem $\pi$  on graphs (e.g., finding a minimum dominating set), one defines two distinct distributed languages: the language $\opt_\pi$ (resp., $\adm_\pi$) is composed of all configurations $(G,x)$ such that $x$ encodes an optimal (resp., admissible) solution for $\pi$ in graph $G$. The minimum-weight spanning tree (MST) problem, which is one of the most studied problem in the context of network computing~\cite{KormanK07,KormanKM15,KormanKP10},  is a typical example of optimization problems that we aim at considering in this section, but many other problems such as maximum independent set, max-cut, etc., are also of our interest. We show that, for any optimization problem $\pi$,  if deciding whether a candidate solution for $\pi$ is admissible is ``easy'',  and if the objective function for $\pi$ has some natural additive form, then $\opt_\pi\in \co{1}$, and thus  $\opt_\pi\in \Lambda_2$.
\begin{theorem}\label{thm:optimization}
Let $\pi$ be an optimization problem on graphs. If the  following two properties are satisfied: {\rm (a)}~$\adm_\pi\in \Lambda_1\cap \co{1}$, and {\rm (b)} the value to the objective function for $\pi$ is the sum, over all nodes, of an individual value at each node which can be computed locally and encoded on $O(\log n)$ bits, then $\opt_\pi \in \co{1}$.
\end{theorem}

\begin{proof}
Let us first prove the following fact. Suppose that every node $u$ of a graph $G=(V,E)$ is given a value $x_u$ on $O(\log n)$ bits, and a value~$s_u$ also on $O(\log n)$ bits. Checking whether $s_u=\sum_{v\in V}x_v$ for every node~$u$ can be achieved by a $\Lambda_1$-algorithm. We describe the algorithm, with, once again, certificates based on a spanning tree. Given a (rooted) spanning tree $T$, every node $u$ is given the certificate for $T$, along with the weight $\sum_{v \in V(T_u)}x_v$ of its subtree~$T_u$. The node~$u$ checks that (1) the spanning tree certificates are locally correct, (2) its value $s_u$ is equal to $s_v$ for each neighbour~$v$, and (3) that the given weight of its subtree is the  sum of the weights of the subtrees rooted at its children, plus~$x_u$. The root~$r$ also checks that the weight of the entire tree is equal to the given value~$s_r$. If one of these properties does not hold at some node, then that node rejects. It follows that every node accepts if and only if $s_u=\sum_{v\in V}x_v$ for every node~$u$. Note that this \emph{gathering technique} can be extended to  functions different from the sum, such as min or max.

Now, let  $\pi$  be an optimization problem on graphs satisfying the conditions of the theorem. To prove $\opt_\pi\in\co{1}$, we show that $\overline{\opt_\pi}\in \Lambda_1$. We describe what certificates are assigned to the nodes by the prover, given an instance $(G,x) \in \overline{\opt_\pi}$. Note that such instance may satisfy either $x$ is not admissible in $G$, or  $x$ is admissible but not optimal. 

\begin{itemize}
\item If $(G,x)\notin\adm_\pi$, then the prover flags each node with $\bot$, and assigns certificates for proving that $(G,x)\notin\adm_\pi$, which is possible thanks to Condition~(a).
\item Otherwise, i.e.,  $(G,x)\in\adm_\pi\setminus \opt_\pi$,  the prover flags each node with $\top$, and assigns certificates  for proving that $(G,x)\in\adm_\pi$ and $(G,x')\in\adm_\pi$, where $x'$ is an arbitrary optimal solution. The latter two sets of certificates can be assigned thanks to Condition~(a). Finally, the prover assigns certificates using the gathering technique to certify the values of the objective function for both $x$ and $x'$. 
\end{itemize}

The nodes then check that all these certificates are consistent, and, in the case with flag $\top$, that  indeed the objective function for $x'$ is better than the one for $x$. If any of these conditions does not hold at some node, then that node rejects.  As a consequence, all the nodes accept if and only if $(G,x)\notin\opt_\pi$. Thus $\overline{\opt_\pi}\in\Lambda_1$.
\end{proof}

Let us give concrete examples of problems satisfying hypotheses (a) and (b). In fact, most classical optimization problems are satisfying these hypotheses, and all the ones typically investigated in the framework of local computing (cf. the survey~\cite{Suomela13}) do satisfy~(a) and~(b).

\begin{corollary}\label{coro:optimization_problems} 
Let $\pi$ be one of the following optimization problems: maximum independent set, minimum dominating set, maximum matching, max-cut, or min-cut.  Then $\opt_\pi \in \co{1}$.
\end{corollary}

\begin{proof}
In view of Theorem~\ref{thm:optimization}, it is sufficient to show that Conditions~(a) and~(b) are satisfied by each problem in this list.  Each of the problems maximum independent set, minimum dominating set, and maximum matching has an easy encoding: a bit that has value 1 if the node is in the set, and zero otherwise. 
These problems satisfies $\adm_\pi\in\Lambda_0$, because each node can check that the local condition specifying admissible solutions holds. It follows that Condition~(a) is satisfied for all these three optimization problems. 
The two cut problems are even easier. The input is a bit that describe on which part of the cut the node is, and every input is admissible as every partition of the nodes defines a cut.  

Regarding Condition~(b), the objective function of maximum independent set, as well as of minimum dominating set, is just the sum of a 0-1 function at each node (0 if not in the set, and 1 otherwise). For maximum matching, as well as for both cut problems, the objective function can be defined as the sum, over all nodes, of half the number of edges adjacent to the node that are involved in the solution. 
\end{proof}

The following other corollary of Theorem~\ref{thm:optimization} deals with two specific optimization problems, namely travelling salesman and MST. The former illustrates a significant difference between the local hierarchy defined from distributed graph automata in~\cite{Reiter15}, and the one in this paper. Indeed, we show that travelling salesman is at the second level of our hierarchy, while it does not even belong to the  graph automata hierarchy (as Hamiltonian cycle is not in MSO). Let $\ts$ be the distributed language formed of all configurations $(G,x)$ where $G$ is a weighted graph, and $x$ is an Hamiltonian cycle $C$ in $G$ of minimum weight (i.e., at node~$u$, $x(u)$ is the pair of edges incident to~$u$ in~$C$). Similarly, let $\mst$ be the distributed language formed of all configurations $(G,x)$ where $G$ is a weighted graph, and $x$ is a MST $T$ in $G$ (i.e., at node~$u$, $x(u)$ is the parent of~$u$ in~$T$). Note that the case of MST is also particularly interesting. Indeed, $\mst$ is known to be in LCP($\log^2(n)$)~\cite{KormanK07}, but  not in $\Lambda_1=\mbox{LCP}(\log(n))$~\cite{KormanKM15}. Note also that, for $\mst$, it is possible to trade locality for the size of the certificates, as it was established in~\cite{KormanKM15} that one can use logarithmic certificates to certify $\mst$ in a logarithmic number of rounds.  A particular consequence of Theorem~\ref{thm:optimization} is the following. 

\begin{corollary} \label{coro:travelsalesman}
$\mst\in\co{1}$ and $\ts \in \co{1}$ for weighted graphs with  weights bounded by a polynomial in~$n$.
\end{corollary}

\begin{proof}
As for Corollary~\ref{coro:optimization_problems}, it is sufficient to prove that both languages satisfy the two conditions of Theorem~\ref{thm:optimization}. Condition~(b) is satisfied for both as their objective function can be defined as the sum, over all nodes, of half the sum of the weights of the edges incident to the node that are involved in the solution (which can be stored on $O(\log n)$ bits as long as all weights have values polynomial in~$n$). 

To prove Condition~(a), we need to show that checking whether a collection $C$ of edges is an Hamiltonian cycle (resp., is a spanning tree) is in $\Lambda_1\cap \co{1}$.  We already noticed earlier in the paper that $\stt\in\Lambda_1$. To prove that $\hc\in\Lambda_1$, we describe a protocol for that language. Given an Hamiltonian cycle~$C$, the prover elects an arbitrary node $r$ of $C$ as a root, picks a spanning tree $T$ rooted at~$r$, and orients $C$ to form a 1-factor (each node has out-degree~1 in $C$). The certificate at node~$u$ is the identity of~$r$, the distance from $u$ to $r$ in $C$ traversed according to the chosen orientation, and the certificate for~$T$. The verification algorithm at node~$u$ checks the tree $T$ (including whether $u$ agrees with all its neighbours on the identity of~$r$). Node $u$ also checks that one of its neighbours in $C$ is one hop closer to $r$, while the other neighbour in $C$ is one hop farther away from $r$. Node $r$ checks that one of its neighbours in $C$ is at distance 1 from it in $C$, while the other it at some distance $>1$. Also, a node with distance~0 in $C$ checks that it is the root of $T$. If all these tests are passed, then node~$u$ accepts, otherwise it rejects. The check of $T$ guaranties that there is a unique node $r$. The check of the hop distance along $C$ guaranties that all nodes are on the same cycle. If both checks are satisfies then $C$ is a unique cycle, covering all nodes, and therefore $C$ is an Hamiltonian cycle. 

Now, it remains to prove that $\stt\in\co{1}$ and $\hc\in\co{1}$.  Let $F$ be a collection of edges that is not forming a spanning tree of $G$. The following certificates are assigned to the nodes. If $F$ is not spanning all nodes, then let $r$ be a node not spanned by $F$, and let $T$ be a spanning tree rooted at~$r$. The certificate of every node $u$ is a pair $(f(u),c(u))$ where the flag $f(u)=0$, and $c(u)$ is the certificate for $T$. If $F$ is spanning all nodes, but contains a cycle $C$, then let $r$ be a node of $C$, let us orient the edges of $C$ in a consistent manner, and let $T$ be a spanning tree rooted at~$r$. The certificate at node $u$ is a pair $(f(u),c(u))$ where $f(u)=1$, and $c(u)$ is a certificate for $T$. In addition, if $u$ belongs to $C$, then $u$ received as part of its certificate its distance to $r$ in the oriented cycle $C$. Finally, If $F$ is spanning forest, then let $F=\{T_1,\dots,T_k\}$ be the trees in $F$, with $k\geq 2$, and root each one at an arbitrary node $r_i$, $i=1,\dots,k$. Let $T$ (resp., $T'$) be a spanning tree rooted at $r_1$ (resp., $r_2$). For every $i\in\{1,\dots,k\}$, the certificate at node $u$ of $T_i$ is a 4-tuple  $(f(u),\idx(u),c(u),c'(u))$ where $f(u)=2$, $\idx(u)=i$, and $c(u)$ (resp., $c'(u)$) is the certificate for $T$ (resp., $T'$).

The verification procedure is as follows. All nodes checks that they have the same flag $f$. A node detecting that flags differ rejects. A node with flag~0 checks the tree certificates, and the root of the tree checks that it is not spanned by $F$. A node with flag~1 checks the tree certificates, and the root $r$ of the tree checks that it belongs to $C$ (i.e., was given distance~0 on the cycle). The root $r$ also checks that it has one neighbours in $C$ at distance~1, and another neighbour at distance $>1$. All other nodes on $C$ check consistency of the distance counter. Finally, a node with flag~2 checks its tree certificates. The root of $T$ checks that it is of index~1, while the root of $T'$ checks that it is of index~2.  Moreover, every node checks that its incident edges in $F$ have extremities with same index. In the three cases, if all tests are passed, the node accepts, otherwise it rejects. 

By construction, if $F$ is not a spanning tree, then all nodes accepts. Instead, if $F$ is a spanning tree, then certificates with different flags cannot yield all nodes to accept since two adjacent nodes with different flags both reject. A flag~0 cannot yield all nodes to accept because the non spanned node does not exist. Similarly, a flag~1 cannot yield all nodes to accept because the cycle does does not exist, and a flag~2 cannot yield all nodes to accept because there are no two different connected components, and hence no two trees $T$ and $T'$ rooted at nodes with different indexes. 

The proof of  $\hc\in\co{1}$ proceeds similarly. Therefore, both Conditions (a) and (b) are satisfied for both $\mst$ and $\ts$, and the result follows by Theorem~\ref{thm:optimization}. 
\end{proof}

\subsection{Non-trivial automorphism}

The graph automorphism problem is the problem of testing whether a given graph has a nontrivial automorphism (i.e., an automorphism\footnote{Recall that $\phi:V(G)\to V(G)$ is an automorphism of $G$ if and only if $\phi$ is a bijection, and, for every two nodes $u$ and $v$, we have: $\{u,v\} \in E(G)\iff \{\phi(u),\phi(v)\}\in E(G)$.} different from the identity). Let $\nta$ be the distributed  language composed of the (connected) graphs that admit such an automorphism. It is known that this language is maximally hard for locally checkable proofs, in the sense that it requires proofs with size $\Omega(n^2)$ bits~\cite{GoosS11}. Nevertheless, we prove that this language remains relatively low in the local hierarchy.

\begin{theorem}
$\nta\in\Lambda_3$.
\end{theorem}

\begin{proof}
The first label $\ell_1$ at node $u$ is an integer that is supposed to be the identity of the image of $u$ by a nontrivial automorphism. Let us denote by  $\phi:V(G)\to V(G)$ the mapping induced by~$\ell_1$. We are left with proving that deciding whether a given $\phi$ is a nontrivial automorphism of $G$ is in $\Lambda_2$. Thanks to Theorem~\ref{thm:cotheorem}, it is sufficient to prove that this decision can be made in co-$\Lambda_1$. Thus let us prove that checking that $(G,\phi)$ is \emph{not} a nontrivial automorphism is in $\Lambda_1$. If $\phi$ is the identity, then the certificate can just encode a flag with this information, and each node $u$ checks that $\phi(u)$ is equal to its own ID. So assume now that $\phi$ is distinct from the identity, but is not an automorphism. To certify this, the prover assigns to each node a set of at most four spanning tree certificates, that ``broadcast'' to all nodes the identity of at most four nodes witnessing that $\phi$ is not an automorphism. Specifically, if $\phi(u)=\phi(v)$ with $u\neq v$, then the certificates are for three spanning trees, respectively rooted at $u,v$, and $\phi(u)$, and if $\{u,v\} \in E(G)$ is mapped to $\{\phi(u),\phi(v)\}\notin E(G)$, or $\{u,v\} \notin E(G)$ is mapped to $\{\phi(u),\phi(v)\}\in E(G)$, then the certificates are for four spanning trees, respectively rooted at $u,v, \phi(u)$, and $\phi(v)$. Checking such certificates can be done locally, and thus checking that $(G,\phi)$ is \emph{not} a nontrivial automorphism is in $\Lambda_1$, from which it follows that $\nta \in \Lambda_3$.  
\end{proof}

\subsection{Problems from the polynomial hierarchy}

As the local hierarchy $\lh{}$ is inspired by the polynomial hierarchy, it is natural to ask about the existence of  connections between their respective levels. In this section, we show that some connections can indeed be established, for central problems in the polynomial hierarchy. For instance, let $k\geq 0$, and let us consider all (connected) graphs $G=(V,E)$ such that there exists $X\subseteq V$, $|X|\geq k$, such that, for every $S\subseteq X$, there is a cycle $C$ in $G$ containing all vertices in $S$, but none in $X\setminus S$. Such graphs have Cycle-VC-dimension, $\mbox{\rm VC}_{cycle}(G)$, at least $k$. Deciding whether, given $G$ and $k$, we have $\mbox{\rm VC}_{cycle}(G)\geq k$ is $\Sigma_3^P$-complete~\cite{Schaefer99,SchaeferU02}. Let $\vck$ be the distributed language composed of all configurations $(G,k)$ such that all nodes of $G$ have the same input~$k$, and $\mbox{\rm VC}_{cycle}(G)\geq k$. 

\begin{theorem}
$\vck\in\Lambda_3$.
\end{theorem}

\begin{proof}
The existence of the set $X$ can be certified setting a flag at each node in $X$, together with a  tree $T_X$ spanning $X$ for proving that $|X|\geq k$. Given $S\subseteq X$, the cycle $C$ can be certified in the same way as the Hamiltonian cycle in the proof of Corollary~\ref{coro:travelsalesman}. 
\end{proof}

Recall that QBF-$k$-SAT is the problem of whether a formula of the type 
\[
\exists y_1, \forall y_2, ..., Q y_k \Phi(y_1,...,y_k),
\]
can be satisfied, where the $y_i$'s are sets of literals on (distinct) boolean variables, $\Phi$ is formula from propositional logic (we can assume, w.l.o.g., that $\Phi$ is in conjunctive normal form), and $Q$ is the universal quantifier if $k$ is even, and the existential quantifier otherwise. The literals in $y_i$ are said to be at the $i$th level. This problem is complete for the $k$th level $\Sigma_k^P$ of PH. It can be rephrased equivalently into a graph problem, by defining the distributed language $\ksat_k$ formed of all configurations $(G,x)$ where $V(G)=C\cup L$, where $C$ is for clauses,  and $L$ is for literals, and there is edge between the positive and negative literals of a same variable, as well as an edge between each clause and all the literals appearing in the clause. More precisely, the input $x(u)$ of a node $u$ in $G$ can be of the form $(\ell,i,s)$ where $\ell$ stands for ``literal'', $i\in\{1,\dots,k\}$ is the level of that literal, and $s\in\{+,-\}$ indicates whether the literal is positive or negative, or of the form $(c)$ where $c$ stands for ``clause''. There is an edge between $(\ell,i,s)$ and $(\ell,i,\bar{s})$ for all literals, and there is an edge between each clause node (i.e., labelled $(c)$) and all the nodes $(\ell,i,s)$ such that the corresponding literal appears in that clause. We set $(G,x)\in \ksat_k$ if and only if the corresponding formula is in QBF-$k$-SAT. 
 
\begin{theorem}\label{theo:ksatinlambdak}
$\ksat_k \in\Lambda_k$.
\end{theorem}

\begin{proof}
For a configuration in $\ksat_k$, the certificates are given to the nodes in the natural way, assigning their values to the literals at odd levels. Each literal node can locally check that its value is the opposite of the one given to its negation, and each clause node can locally check that it is linked to at least one literal that has value true.
\end{proof}

\subsection{Connections to descriptive complexity}

The notion of locality is an important subject when considering expressibility of different logics on graphs, as illustrated by the locality result of Schwentick and Barthelmann~\cite{SchwentickB99} for first-order logic. It is then a natural question to ask which are the logics that express properties in $\lh{}$. A first answer was given by Göös and Suomela~\cite{GoosS11} who proved that all properties expressible in existential-MSO are in $\loglcp{}$. One can then expect that $\lh{}$ contains MSO, and it is indeed the case. An easy way to prove this fact is to use the recent result of Reiter \cite{Reiter15}. 

\begin{theorem}\label{theo:msoinlh}
$\mso\in\lh$.
\end{theorem}

\begin{proof}
As we pointed out earlier in the text, distributed graph automata (DGA) are based on a combination of hypotheses, some weaker  than the LOCAL model, and other ones stronger. On the one hand, all computation and communication steps in DGA can be simulated in the LOCAL model. On the other hand, the decision mechanism in DGA is stronger than the one typically used in local decision, as far as distributed network computing is concerned. Specifically, in our framework, the nodes output true or false (i.e., 1 or 0), and the instance is accepted if and only if all the outputs are true. That is, the decision mechanism is simply the conjunction of all the outputs, whereas, in DGA, the outputs belong to an arbitrary finite set $S$, and the decision mechanism is an arbitrary function $f$ from the set $\mathcal{O}$ of the outputs to $\{\mbox{accept, reject}\}$. Note that $\mathcal{O}$ is a \emph{set}, and therefore a same output at two different nodes appears only once in $\mathcal{O}$.
 
Let us consider a language $L$ at level $k$ of the DGA hierarchy. We show that this language is at level at most $k+1$ of $\lh{}$. Indeed we can run exactly the same protocol as in DGA, but the decision decision mechanism. Nevertheless, the decision mechanism can be simulated with an additional existential quantifier certifying a spanning tree $T$ that is used to gather all the outputs of the nodes produced by the DGA protocol, in a way similar to the one in Claim~\ref{claim:spanningtree}: each node checks that its set of outputs is the union of the output sets of its children in $T$. The root of $T$ stores the entire set $\mathcal{O}$ (which can be done using $O(|S|\log |S|)=O(1)$ bits), computes $f(\mathcal{O})$, and accepts or rejects accordingly. 
\end{proof}

\section{Conclusion}

In this paper, we have defined and analyzed a local hierarchy $\lh$ of decision generalizing proof-labelling schemes and locally checkable proofs. Using this hierarchy, we have defined interactive local decision protocols enabling to decrease the size of the distributed certificates. We have defined the hierarchy for $O(\log n)$-bit size labels, mostly because this extends the class $\llcp$ in~\cite{GoosS11}, and because this size fits with the classical CONGEST model for distributed computation~\cite{Peleg00}. However, most of our results can be extended to labels on $O(B(n))$ bits, for $B(n)$ larger than $\log n$. In particular, it is worth noticing that the existence of a language $L$ outside $\lh$  holds as long as $B=o(n)$. 

The main open problem is whether $\lh$ has infinitely many levels, or whether it collapses at some level $\Lambda_k$. We know that the latter can only happen for $k\geq 2$, and thus it would be quite interesting to know whether $\Lambda_3\neq \Lambda_2$. In particular, all the typical counting arguments used to separate $\Lambda_2$ from $\Lambda_1$, or, more generally, to give lower bounds on the label size in proof-labelling schemes or locally checkable proofs appear to be too weak for separating $\Lambda_3$ from $\Lambda_2$. A separation result for $\Lambda_3\neq \Lambda_2$ would thus probably provide new tools and concepts for the design of space lower bounds in the framework of distributed computing. 

\subparagraph*{Acknowledgements} 
We thank Jukka Suomela for pointing out that a counting argument can be used to find languages outside the local hierarchy, and Fabian Reiter for fruitful discussions about distributed graph automata. 

\bibliographystyle{plain}
\bibliography{bibliography}

\end{document}